\documentclass[12pt]{article}
\usepackage{mathrsfs}
\usepackage{dsfont}
\usepackage{amsthm}
\usepackage{mathrsfs}
\usepackage{amsmath}
\usepackage{amsfonts}
\usepackage[colorlinks,linkcolor=black,anchorcolor=blue,citecolor=blue]{hyperref}
\usepackage{amssymb, amsmath, cite}

\newtheorem{theorem}{Theorem}[section]
\newtheorem{lemma}{Lemma}[section]

\newcommand\be{\begin{equation}}
\newcommand\ee{\end{equation}}
\newcommand\ber{\begin{eqnarray}}
\newcommand\eer{\end{eqnarray}}
\newcommand\berr{\begin{eqnarray*}}
\newcommand\eerr{\end{eqnarray*}}
\newcommand\bea{\begin{eqnarray}}
\newcommand\eea{\end{eqnarray}}

\newcommand{\nn}{\nonumber}

\newcommand{\dd}{\mathrm{d}}
\newcommand\e{\mathrm{e}}\newcommand\pa{\partial}
{\setlength\arraycolsep{2pt}

\begin{document}

\title{Cosmic strings in a generalized linear formulation of gauge field theory}
\author{Lei Cao\\School of Mathematics and Statistics\\Henan University\\
Kaifeng, Henan 475004, PR China\\ \\Shouxin Chen\\
School of Mathematics and Statistics\\Henan University\\
Kaifeng, Henan 475004, PR China}
\date{}
\maketitle

\begin{abstract}
In this note we construct self--dual cosmic strings from a gauge field theory with a generalized linear formation of potential energy density. By integrating the Einstein equation, we obtain a nonlinear elliptic equation which is equal with the sources. We prove the existence of a solution in the broken symmetry category on the full plane and the multiple string solutions are valid under a sufficient condition imposed only on the total string number $N$. The technique of upper--lower solutions and the method of regularization are employed to show the existence of a solution when there are at least two distant string centers. When all the string centers are identical, fixed point theorem are used to study the properties of the nonlinear elliptic equation. Finally, We give the sharp asymptotic estimate for the solution at infinity.
\end{abstract}

\medskip
\begin{enumerate}

\item[]

\item[]
{MSC numbers(2020):} 35J60, 81T13.

\end{enumerate}

\section{Introduction}\label{s1}

Spontaneous symmetry breaking causes phase transitions in the early universe and generates stable topological defects, such as domain walls, vortex strings, and monopoles \cite{Ki,Kib}. These objects plays important roles in cosmology and has aroused great concern in recent years. Particularly, cosmic strings have received the considerable attention for the reason that the wide variety of astrophysical phenomena \cite{Ze,Vi,Wi,Br} they have been involved to explain.

Static solutions of the coupled Einstein and Yang--Mills--Higgs equations are called cosmic strings, which keep cylindrical symmetry. And the system represents the coupling of gravity and superconductivity under the action of a gauged Abelian group. The multi--center display of energy and curvature around the local strings is one of the most interesting features of the cosmic string solutions. These local strings are believed to seed the accretion of matter in the early stages of galaxy formation.

Based on a prescribed Dirac distribution source formalism, Letelier \cite{Le} gave the first and simplest example of a multiple string solution to the reduced Einstein equation. The second concrete example of cosmic strings that lay structural foundation is from the construction of Comtet and Gibbons \cite{Co} based on the $\sigma$--model or the harmonic map model. Yang \cite{YaO,YaP} constructed multiple cosmic strings of classical Abelian Higgs theory based on the formalism in \cite{Co}. When gravity is absent, the Hamiltonian energy density describing multiply distributed electrically and magnetically charged vortices, that is to say, cosmic string solutions degenerate into vortex solutions without gravity. Vortex solutions in Chern--Simons--Higgs theory have been well studied in \cite{Caf,Cha,Spr,Spru,Tar}. See also Manton \cite{Ma} and Gudnason \cite{Gu} for some recent reviews on vortex equations. The original goal of our work is to investigate the existence of string solutions introduced from an gauge field model \cite{Lo,Loh,Ho} containing the Higgs field $\phi$ which is a cross--section on the $U(1)$--line bundle $L$ \cite{Gr} over the spacetime and the gauge field $A$ which is a connection $1$--form. We transform the self--dual system and Einstein equation into a string equation which is a nonlinear elliptic equation with highly complex exponential form. Then, we use the technique of upper--lower solutions, the method of regularization and fixed point theorem to establish the existence theorem of solutions of the nonlinear elliptic equation.

The rest of our paper is organized as follows. In section 2, we introduce the gauge field model to be used here to generate a system of strings and reduce the equation of motion into a self--dual system. Section 3 is devoted to obtain an elliptic equation equal with the self--dual system and Einstein equation. In section 4, we prove the existence of a solution for the elliptic equation governing strings on the full plane. In order to avoid the difficulty caused by gravity, we introduce a $\delta$--regularization of the equation and then take limit as $\delta\to 0$ to solve the original equation. According to the distribution of string centers, the proof is divided into two cases.

\section{Gauge field theory and strings}\label{s1}

According to \cite{Lo,Loh,Ho} the Lagrangian is
\be\label{1.1}
\mathcal{L}=-\frac{1}{4}F_{\mu\nu}^2+\frac{1}{2}F(s)D_\mu\phi\overline{D_\nu\phi}-\frac{1}{2}w^2(s),
\ee
where $F_{\mu\nu}=\partial_\mu A_\nu-\partial_\nu A_\mu$ is the electromagnetic curvature induced from the 4-vector connection $A_\mu (\mu, \nu=0,1,2,3, t=x_0)$, $D_\mu\phi=\partial_\mu\phi-iA_\mu\phi$ is the gauge--covariant derivatives, $F$ and $w$ are real functions, depending only on $s=|\phi|^2$. Let $g_{\mu\nu}$ be the gravitational metric of the spacetime of signature $(+---)$. Thus the action density of the matter--gauge sector that modifies \eqref{1.1} into
\be\label{1.2}
\mathcal{L}=-\frac{1}{4}g^{\mu\alpha}g^{\nu\beta}F_{\mu\nu}F_{\alpha\beta}+\frac{1}{2}g^{\mu\nu}F(s)D_\mu\phi\overline{D_\nu\phi}-\frac{1}{2}w^2(s),
\ee
The energy--momentum tensor obtained by varying the gravitational metric in the action $L=\int\mathcal{L}$ with $\mathcal{L}$ being defined by \eqref{1.2} is
\be\label{1.4}
T_{\mu\nu}=-g^{\alpha\beta}F_{\mu\alpha}F_{\nu\beta}+\frac{1}{2}F(s)(D_\mu\phi\overline{D_\nu\phi}+\overline{D_\mu\phi}D_\nu\phi)-g_{\mu\nu}\mathcal{L}.
\ee
We are interested in the Einstein equations coupled with the Higgs model \eqref{1.2}, that is
\bea
&&G_{\mu\nu}=-8\pi G T_{\mu\nu},\label{1.3a}\\
&&\frac{1}{\sqrt{|\text{det}(g_{\alpha\beta})|}}\pa_\mu\left(F(s)\sqrt{|\text{det}(g_{\alpha\beta})|}g^{\mu\nu}\pa_\nu \phi\right)=\left(F'(s)g^{\mu\nu}D_\mu u\overline{D_\nu \phi}-2w(s)w'(s)\right)\phi,\label{1.3b}\\
&&\frac{1}{\sqrt{|\text{det}(g_{\alpha\beta})|}}\pa_{\mu'}\left(g^{\mu\nu}g^{\mu'\nu'}\sqrt{|\text{det}(g_{\alpha\beta})|} F_{\nu\nu'}\right)=F(s)\frac{i}{2}g^{\mu\nu}(\phi\overline{D_\nu \phi}-\overline{\phi}D_\nu \phi),\label{1.3c}
\eea
which produced cosmic string solutions.
Cosmic string solutions are a special class of static field configuration which depend on the coordinate $(x_j)(j=1,2)$, $A_0=A_3=0$, and the metric is of the form
\be\label{1.5}
g_{\mu\nu}=\text{diag}\{1,-\e^\eta,-\e^\eta,-1\}.
\ee
Then $T_{\mu\nu}$ is simplified to
\bea\label{1.6}
T_{00}&=&\mathcal{H}, T_{33}=-\mathcal{H}, T_{03}=T_{0j}=T_{3j}=0,\notag\\
T_{jk}&=&g^{\alpha\beta}F_{j\alpha}F_{k\beta}+\frac{1}{2}F(s)(D_j\phi\overline{D_k\phi}+\overline{D_j\phi}D_k\phi)-g_{jk}\mathcal{H},
\eea
where
\be\label{1.7}
\mathcal{H}=\frac{1}{4}g^{\mu\alpha}g^{\nu\beta}F_{\mu\nu}F^{\mu\nu}+\frac{1}{2}g^{\mu\nu}F(s)D_\mu\phi\overline{D_\mu\phi}+\frac{1}{2}w^2(s)
\ee
is the energy density of the Higgs sector. Besides, if the Gauss curvature of $(\mathbb{R}^2,\e^\eta\delta_{jk})$ be denoted by $K_g$ which can be calculated by the expression
\be\label{1.8a}
K_g=-\frac 12 \e^{-\eta}\triangle\eta,
\ee
and the Einstein tensor reduces into the form
\bea\label{1.8}
-G_{00}&=&G_{33}=K_g,\nn\\
G_{\mu\nu}&=&0~~\text{for other values of}~~\mu, \nu.
\eea
Then, the Einstein equations \eqref{1.3a} become the following two--dimensional Einstein equations
\be\label{1.9}
K_g=8\pi G\mathcal{H}.
\ee

The system \eqref{1.3a}--\eqref{1.3b} is hard to approach. However, we can reduce the difficulty by explore the self--dual structure. To proceed, we introduce a current density
\be\label{1.10}
J_k=\frac{i}{2}f(s)(\phi \overline{D_k\phi}-\overline{\phi} D_k\phi),~~j=1,2.
\ee
By Kato's identity
\be\label{1.11}
\partial_j(\phi\overline{\psi})=\overline{\psi}\partial_j\phi+\phi\partial_j\overline{\psi}=\overline{\psi}D_j\phi+\phi \overline{D_j\psi},
\ee
we can easily verify the following commutation relation for gauge--covariant derivatives
\be\label{1.12}
[D_j,D_k]\phi=(D_j D_k-D_k D_j)\phi=-i(\partial_j A_k-\partial_k A_j)\phi=-i F_{jk}\phi,~~j,k=1,2.
\ee
Differentiating \eqref{1.10} and applying $\partial_j(\phi\overline{\phi})=\overline{\phi}D_j\phi+\phi \overline{D_j\phi}$, we obtain
\be\label{1.13}
J_{12}=\partial_1 J_2-\partial_2 J_1=-sf(s)F_{12}+i(f(s)+f'(s)s)(D_1\phi \overline{D_2\phi}-\overline{D_1\phi} D_2\phi).
\ee
Besides, it can be shown that there holds the identity
\be\label{1.14}
|D_1\phi\pm i D_2\phi|^2=|D_1\phi|^2+|D_2\phi|^2\mp i(D_1\phi \overline{D_2\phi}-\overline{D_1\phi} D_2\phi).
\ee

Therefore, in order to derive the self--dual reduction of the system \eqref{1.3b}, we may apply \eqref{1.10}, \eqref{1.13}, and \eqref{1.14} to note that the energy density $\mathcal{H}=T_{00}$ induced from the energy--momentum tensor \eqref{1.4} of the matter--gauge sector has the representation
\bea\label{1.15}
\mathcal{H}&=&\frac 12\e^{-2\eta}F_{12}^2+\frac 12\e^{-\eta}F(s)(|D_1\phi|^2+|D_2\phi|^2)+\frac 12 w^2(s)\nn\\
&=&\frac 12(\e^{-\eta}F_{12}\mp w)^2+\frac 12\e^{-\eta}F(s)|D_1\phi\pm i D_2\phi|^2\pm \e^{-\eta}F_{12}\nn\\
&& \pm \e^{-\eta}(w-1)F_{12}\pm \frac i2\e^{-\eta}F(s)(D_1\phi \overline{D_2\phi}-\overline{D_1\phi} D_2\phi)\nn\\
&=&\frac 12(\e^{-\eta}F_{12}\mp w)^2+\frac 12\e^{-\eta}F(s)|D_1\phi\pm i D_2\phi|^2\pm \e^{-\eta}(F_{12}+J_{12}).
\eea
Thus, we obtain the self--dual system
\bea
D_1\phi\pm i D_2\phi&=&0,\label{1.16}\\
F_{12}&=&\pm\e^{\eta}w(s).\label{1.17}
\eea
It is easy to examine that any solution of \eqref{1.16}--\eqref{1.17} necessarily satisfies \eqref{1.3b} where the metric tensor is defined by \eqref{1.5}. Besides, comparing the terms in \eqref{1.15} we can get the following equivalence relations
\bea
-f(s)s&=&w(s)-1,\label{1.18}\\
f(s)+f'(s)s&=&\frac 12 F(s).\label{1.19}
\eea
Differentiating \eqref{1.18} we have
\be\label{1.20}
-f(s)-f'(s)s=w'(s),
\ee
then in view of \eqref{1.19}, we find that
\be\label{1.21}
w(s)=\frac 12\int_s^1 F(t)\dd t.
\ee

\section{The governing elliptic problem}\label{s1}

In this section we derive an nonlinear elliptic equation from the self--dual system \eqref{1.16}--\eqref{1.17} and the Einstein equation \eqref{1.9}. By the equivalent, we can get the multi--string solutions through solving the nonlinear elliptic equation. For this purpose, use the equation \eqref{1.16}, we obtain
\be\label{1.22}
(\pa_1\pm i \pa_2)\ln\phi=i(A_1\pm i A_2).
\ee
Thus
\be\label{1.23}
F_{12}=\pa_1A_2-\pa_2A_1=\mp\frac 12 \triangle\ln|\phi|^2,
\ee
and the system \eqref{1.16}--\eqref{1.17} is reduced via $u=\ln |\phi|^2$ to
\be\label{1.24}
\triangle u=-2\e^{\eta}w(\e^u)+4\pi\sum_{s=1}^N \delta_{p_s},
\ee
where $p_s, s=1,\cdots,N$ are zeros of $u$ and determine the locations of strings, $\delta_{p_s}$ is the Dirac distribution at $p_s$.

We are now ready to resolve the Einstein equation \eqref{1.9}.
In view of \eqref{1.17}--\eqref{1.19} and \eqref{1.23}, when away from zeros of $u$, we may rewrite the energy density $\mathcal{H}$ in the form
\bea\label{1.25}
\mathcal{H}&=&\frac 12\e^{-2\eta}F_{12}^2+\frac 12\e^{-\eta}F(s)(|D_1\phi|^2+|D_2\phi|^2)+\frac 12 w^2(s)\nn\\
&=&\e^{-2\eta}F_{12}^2+\frac 12\e^{-\eta}F(s)(|D_1\phi|^2+|D_2\phi|^2)\nn\\
&=&\pm\e^{-\eta}F_{12}w+\frac 12\e^{-\eta}F(s)\left(\frac 12\e^u|\nabla u|^2\right)\nn\\
&=&\frac 12\e^{-\eta}\left((\e^u f(\e^u)-1)\triangle u+(f(\e^u)+f'(\e^u)\e^u)e^u|\nabla u|^2\right)\nn\\
&=&\frac 12\e^{-\eta} \triangle \left(g(e^u)-u\right),
\eea
where
\be\label{1.26}
g(s)=\int_0^s f(\rho)\dd \rho.
\ee
Since $\mathcal{H}$ is a smooth function, the above expression indicates that we can compensate the singular sources at the points $p_s$ to arrive the relation in the full $\mathbb{R}^2$
\bea\label{1.27}
\e^{\eta}\mathcal{H}&=&\frac 12\triangle\left(g(e^u)-u\right)+2\pi\sum_{s=1}^N\delta_{p_s}\nn\\
&=&\frac 12\triangle\left(g(e^u)-u+\sum_{s=1}^N\ln|x-p_s|^2\right).
\eea
Thus, inserting the above equation \eqref{1.27} and \eqref{1.8a} into \eqref{1.9}, we see that
\be\nn
\frac{\eta}{8\pi G}+g(\e^u)-u+\sum_{s=1}^N\ln|x-p_s|^2
\ee
is a harmonic function, which we assume to be a constant. Consequently, we obtain the explicit form of $\e^\eta$ of the metric as follows
\be\label{1.28}
\e^\eta=\frac{g_0}{2}\left(\e^{g(\e^u)-u}\prod_{s=1}^N|x-p_s|^2\right)^{-8\pi G},
\ee
where $g_0>0$ is an arbitrary constant. Inserting \eqref{1.28} into \eqref{1.24}, we get the final governing equation
\be\label{1.29}
\triangle u=-g_0\left(\e^{g(\e^u)-u}\prod_{s=1}^N|x-p_s|^2\right)^{-8\pi G}w(\e^u)+4\pi\sum_{s=1}^N \delta_{p_s}.
\ee

If we take $w(s)=1-s$, then we arrive at the classical Abelian Higgs strings
\be\label{1.30}
\triangle u=g_0\frac{(\e^u-1)\e^{8\pi G u}}{\e^{8\pi G \e^u}}\left(\prod_{s=1}^N|x-p_s|^{-2}\right)^{8\pi G}+4\pi\sum_{s=1}^N \delta_{p_s}.
\ee
In this paper, we concentrate on the generalized situation $w(s)=1-s^m, m>0$ and the corresponding governing equation read as
\be\label{1.31}
\triangle u=g_0\frac{(\e^{mu}-1)\e^{8\pi G u}}{\e^{\frac{8\pi G}{m}\e^{mu}}}\left(\prod_{s=1}^N|x-p_s|^{-2}\right)^{8\pi G}+4\pi\sum_{s=1}^N \delta_{p_s}.
\ee
We are interested in solutions in the broken symmetry category so that $u=0$ at infinity.


\section{Existence of strings}
\setcounter{equation}{0}

According to the distribution of string centers, we give the existence result for multiple strings with two methods. When there are at least two different string centers, we use the technique of upper-lower solutions. When all the string centers are identical, we introduce fixed point theorem.

\begin{theorem}\label{th2.1}
Suppose $8\pi G N<1$ or $8\pi G N=1$ and there are at least two distant string centers, then the equation \eqref{1.31} has a solution which vanishes at infinity.
\end{theorem}

Define the background functions
\bea
u_0&=&\sum_{s=1}^N\ln\left(\frac{|x-p_s|^2}{1+|x-p_s|^2}\right),\label{2.1}\\
w_0&=&\sum_{s=1}^N\ln\left(1+|x-p_s|^2\right).\label{2.2}
\eea
It is clearly that $u_0<0, w_0\geq 0$, and
\bea\label{2.3}
\triangle u_0&=&4\pi\sum_{s=1}^N \delta_{p_s}-\triangle w_0\nn\\
&=&4\pi\sum_{s=1}^N \delta_{p_s}-g,
\eea
where
\be\label{2.4}
g=\triangle w_0=4\sum_{s=1}^N\frac{1}{(1+|x-p_s|^2)^2}>0.
\ee

Let $u=u_0+v$, then \eqref{1.31} becomes
\be\label{2.5}
\triangle v=g_0\frac{(\e^{m(u_0+v)}-1)\e^{8\pi G (u_0+v)}}{\e^{\frac{8\pi G}{m}\e^{m(u_0+v)}}}\left(\prod_{s=1}^N|x-p_s|^{-2}\right)^{8\pi G}+g.
\ee
In order to avoid some technical difficulties, we consider a regularized form of the above equation \eqref{2.5}
\be\label{2.6}
\triangle v=g_0\frac{(\e^{m(u_0^\delta+v)}-1)\e^{8\pi G (u_0^\delta+v)}}{\e^{\frac{8\pi G}{m}\e^{m(u_0^\delta+v)}}}F_\delta(x)+g,~~(0<\delta<1).
\ee
where
\bea
F_\delta(x)&=&\left(\prod_{s=1}^N(\delta+|x-p_s|^2)\right)^{-8\pi G},\nn\\
u_0^\delta&=&\sum_{s=1}^N\ln\left(\frac{\delta+|x-p_s|^2}{1+|x-p_s|^2}\right).\nn
\eea

\begin{lemma}\label{d.1}
The function $v_1^\delta=-u_0^\delta~(0<\delta<1)$ is a supersolution of \eqref{2.6}.
\end{lemma}
\begin{proof}
In fact, we have
\bea
\triangle v_1^\delta&=&-u_0^\delta=-4\delta\sum_{s=1}^N\frac{1}{(1+|x-p_s|^2)^2}+g\nn\\
&<&g=g_0\frac{(\e^{m(u_0^\delta+v_1^\delta)}-1)\e^{8\pi G (u_0^\delta+v_1^\delta)}}{\e^{\frac{8\pi G}{m}\e^{m(u_0^\delta+v_1^\delta)}}}F_\delta(x)+g,~~(0<\delta<1).\nn
\eea
Thus, $v_1^\delta$ is a supersolution as expected.
\end{proof}

\begin{lemma}\label{d.2}
If $8\pi G N\leq 1$, then there a constant $C_0>0$ independent of $0<\delta<\frac 12$ (say) so that
\be\label{2.7}
0>g_0\frac{(\e^{m u_0^\delta}-1)\e^{8\pi G u_0^\delta}}{\e^{\frac{8\pi G}{m}\e^{m u_0^\delta}}}F_\delta(x)+g,~~\text{in}~\mathbb{R}^2
\ee
hold whenever $g_0>C_0$. In other worlds, $v=0$ is a subsolution of \eqref{2.6} for all $\delta$.
\end{lemma}
\begin{proof}
We may rewrite $u_0^\delta$ as
\be\label{2.8}
u_0^\delta=-\sum_{s=1}^N\ln\left(1+\frac{1-\delta}{\delta+|x-p_s|^2}\right).
\ee
The expression \eqref{2.8} implies that $u_0^\delta\rightarrow 0$ uniformly as $r=|x|\rightarrow\infty$. Thus $\e^{mu_0^\delta}\rightarrow 1$ uniformly as $r\rightarrow\infty$. Note that
\be\nn
\e^{m u_0^\delta}-1=\e^{\xi(m u_0^\delta)m u_0^\delta}m u_0^\delta,~~(0<\xi(m u_0^\delta)<1).
\ee
Thus $\e^{\xi(m u_0^\delta)m u_0^\delta}\rightarrow 1$ uniformly as $r\rightarrow\infty$. Hence
\bea
&&\frac{(\e^{m u_0^\delta}-1)\e^{8\pi G u_0^\delta}}{\e^{\frac{8\pi G}{m}\e^{m u_0^\delta}}}F_\delta(x)\nn\\
=&&\frac{\e^{8\pi G u_0^\delta}}{\e^{\frac{8\pi G}{m}\e^{m u_0^\delta}}}\e^{\xi(m u_0^\delta)m u_0^\delta}m u_0^\delta F_\delta(x)\nn\\
=&&-\frac{m\e^{8\pi G u_0^\delta}\e^{\xi(m u_0^\delta)m u_0^\delta}}{\e^{\frac{8\pi G}{m}\e^{m u_0^\delta}}}\left(\sum_{s=1}^N\ln\left(1+\frac{1-\delta}{\delta+|x-p_s|^2}\right)\right)\left(\prod_{s=1}^N(\delta+|x-p_s|^2)\right)^{-8\pi G}\nn\\
=&&-\frac{m\e^{8\pi G u_0^\delta}\e^{\xi(m u_0^\delta)m u_0^\delta}}{\e^{\frac{8\pi G}{m}\e^{m u_0^\delta}}}\left(\sum_{s=1}^N \frac{1}{1+\xi_s}\cdot\frac{1-\delta}{\delta+|x-p_s|^2}\right)\left(\prod_{s=1}^N(\delta+|x-p_s|^2)\right)^{-8\pi G}\nn\\
\equiv && -h_\delta(x),\nn
\eea
where $\xi_s=\theta_s\cdot\left(\frac{1-\delta}{\delta+|x-p_s|^2}\right),~~0<\theta_s<1,~~~s=1, 2, \cdots, N$.

Note that the assumption $8\pi G N\leq 1$, we have
\be\nn
r^4 h_\delta(x)\rightarrow\infty~~\text{uniformly}~~~(\text{if}~8\pi G N< 1)~~~\text{as}~~r=|x|\rightarrow\infty
\ee
or
\be\nn
r^4 h_\delta(x)\rightarrow \text{some number}~c_\delta>0~~~(\text{if}~8\pi G N= 1)~~~\text{as}~~r=|x|\rightarrow\infty
\ee
for $0<\delta<\frac12$. Therefore, we can find $C_0, r_0$ so that
\be\label{2.9}
-\frac{1}{r^4}(g_0r^4h_\delta-r^4g)<0,~~\forall r=|x|\geq r_0,~~g_0\geq C_0.
\ee
This shows that \eqref{2.7} holds for $r=|x|\geq r_0, g_0\geq C_0$.

Now we choose $r_0$ sufficiently large, such that
\be\nn
\{p_1, p_2, \cdots, p_N\}\subset\left\{x\mid|x|<r_0\right\}\equiv D.
\ee
We see by the definition of $u_0^\delta$ that $u_0^\delta<u_0^{\frac12}~(\delta\leq\frac 12)$, we find
\be\nn
\e^{u_0^\delta}-1\leq\e^{u_0^{\frac 12}}-1,
\ee
\be\nn
\frac{\e^{8\pi G u_0^\delta}}{\e^{\frac{8\pi G}{m}\e^{m u_0^\delta}}}F_\delta\geq\e^{-\frac{8\pi G}{m}}\left(\prod_{s=1}^N(1+|x-p_s|^2)\right)^{-8\pi G}
\ee
Hence, choosing $C_0$ sufficiently large, we see that \eqref{2.7} holds on $D$ as well.
\end{proof}

Since $v_1^\delta$ is a supersolution and $v=0$ is a subsolution, base on Ni \cite{Ni}, we see that \eqref{2.6} has a smooth solution $v^\delta$ in $\mathbb{R}^2$ satisfying
\be\label{2.10}
-u_0^\delta=v_1^\delta>v^\delta\geq 0,~~~~\text{in}~\mathbb{R}^2.
\ee

We are now going to study the behavior of the family $\{v^\delta\}$ as $\delta\rightarrow 0$.

In view of the definition of $u_0^\delta$, we know that $u_0^{\delta_1}<u_0^{\delta_2}$ for $\delta_1<\delta_2$. Thus $v_1^{\delta_1}>v_1^{\delta_2}$. Particularly, $v_0\equiv-u_0=v_1^0>v_1^\delta$ for all $\delta>0$. Therefore a weaker form of \eqref{2.10} is
\be\label{2.11}
v_0>v^\delta\geq 0,~~~~\text{in}~\mathbb{R}^2.
\ee

Considering the right--hand side of \eqref{2.6}. By the relation of \eqref{2.10}, we have
\be\label{2.12}
\e^{-\frac{8\pi G}{m}\e^{m(u_0^\delta+v^\delta)}}\left(\e^{m(u_0^\delta+v^\delta)}-1\right)<\e^{mv_0}-1.
\ee
Besides, note \eqref{2.11} we get the following bound
\be\label{2.13}
\e^{8\pi G (u_0^\delta+v^\delta)}F_\delta(x)\leq F_\delta(x)<\prod_{s=1}^N|x-p_s|^{-16\pi G}\triangleq f.
\ee
So the right--hand side of \eqref{2.6} has $\delta$--independent upper bound but get singularity at $x=p_s, s=1, \cdots, N$. In order to control the sequence $\{v^\delta\}$ with a $L^p$--estimates, we require
\be\label{2.14}
f\in L_{loc}^p(\mathbb{R}^2)~~~~\text{for some}~p>1.
\ee
By \eqref{2.13}, we see that $8\pi G N<1$ under the require of \eqref{2.14}. When $8\pi G N=1$, we assume there are at least two centers of strings. Thus \eqref{2.14} is still hold. In other words, $8\pi G N\leq1$ is sufficient to ensure \eqref{2.14} holds. In conclusion, we know that, for any bounded domain $\mathcal{O}\in\mathbb{R}^2$, there is a $\delta$--independent constant $C(p, \mathcal{O})>0$ such that
\be\label{2.15}
||\triangle v^\delta||_{L^p(\mathcal{O})}\leq C(p, \mathcal{O}).
\ee
Then applying \eqref{2.11} and \eqref{2.15}, and using the interior $L^p$--estimates \cite{Ag,Be}, we get that
\be\label{2.16}
||v^\delta||_{W^{2,p}(\mathcal{O})}\leq C(p, \mathcal{O}).
\ee
Furthermore, according to the continuous embedding
\be\label{2.17}
W^{2,p}(\mathcal{O})\rightarrow C^q(\overline{\mathcal{O}}),~~~\text{for}~0\leq q<2-\frac 2 p,
\ee
we know that $\{v^\delta\}$ is bounded in $C(\overline{\mathcal{O}})$. Combining this fact with \eqref{2.11}, we obtain that $\{v^\delta\}$ is uniformly bounded over the full $\mathbb{R}^2$.

Therefore, in view of \eqref{2.6}, the sequence $\{\triangle v^\delta\}$ is also bounded in $\mathbb{R}^2$. So \eqref{2.16} holds for any $p>1$ and any given bounded domain $\mathcal{O}$ by the interior $L^p$--estimates. Taking $p>2$ in \eqref{2.17}, we have the bound
\be\label{2.18}
|v^\delta|_{C^1(\overline{\mathcal{O}})}\leq C(p, \mathcal{O}),~~\forall \delta>0.
\ee
We see in view of \eqref{2.6} and \eqref{2.18} that
\be\label{2.19}
|\triangle v^\delta|_{C^1(\overline{\mathcal{O}})}\leq C(p, \mathcal{O}),~~\forall \delta>0.
\ee
For the arbitrary of $\mathcal{O}$, \eqref{2.19} and the interior Schauder estimates enable us to conclusion that there is a constant $C(\alpha,p, \mathcal{O})$ independent of $\delta>0$ such that
\be\label{2.20}
|v^\delta|_{C^{2,\alpha}(\overline{\mathcal{O}})}\leq C(\alpha, p, \mathcal{O}),~~\alpha\in(0, 1).
\ee

The above results allow us to use a standard diagonal subsequence argument to show that there is a solution of \eqref{2.5} on $\mathbb{R}^2$ in the limit $\delta\to 0$.

Given a sequence of positive numbers $\{r_i\}$ satisfy $0<r_1<r_2<\cdots<r_i<\cdots$, and $r_i\to\infty$ as $i\to\infty$. Set
\be\nn
B_i=\{x\in\mathbb{R}^2\ \mid |x|<r_i\}.
\ee
Let $\mathcal{O}=B_i(i=1,2,\cdots)$ in \eqref{2.20} and use the compact embedding $C^{2,\alpha}(\overline{B_i})\to C^2(\overline{B_i})$, we see that for each $B_i$, there is a convergent subsequence of $\{v^\delta\}$ in $C^2(\overline{B_i})$. Starting from $B_1$. We can choose $\{\delta_n^1\}$, $\delta_n^1\to\infty$ as $n\to\infty$ and $v_1\to C^2(\overline{B_1})$ satisfying $v^{\delta_n^1}\to v_1$ in $C^2(\overline{B_1})$ as $n\to\infty$. Then there is a subsequence $\{\delta_n^2\}$ of $\{\delta_n^1\}$ and an element $v_2\to C^2(\overline{B_2})$, such that $\delta_n^2\to 0$ and $v^{\delta_n^2}\to v_2$ in $C^2(\overline{B_2})$ as $n\to\infty$. Clearly, $v_1=v_2$ in $B_1$. Repeating the above procedure we finally get the sequences $\{\delta_n^i\}, i=1,2,\cdots$ satisfying

(i) $\{\delta_n^i\}\subset\{\delta_n^{i-1}\}, i=2,3,\cdots$;

(ii) for any fixed $i=1,2,\cdots, \{\delta_n^i\}\to 0$ as $n\to\infty$;

(iii) for any fixed $i=1,2,\cdots$, there exists an element $v_i\in C^2(\overline{B_i})$ such that $v^{\delta_n^i}\to v_i$ as $n\to\infty$;

(iv) there holds $v_i=v_{i-1}$ on $B_i, i=2,3,\cdots$.

Set $v(x)=v_i(x)$ for $x\in B_i$ and $i=1,2,\cdots$. The property (iv) tells us that $v\in C^2(\mathbb{R}^2)$. In view of (i) and (iii), we see that $v^{\delta_n^n}$ convergence to $v$ as $n\to\infty$ in $C^2(\overline{\mathcal{O}})$--norm for any given bounded domain $\mathcal{O}$ in $\mathbb{R}^2$. Rewrite \eqref{2.6} with $\delta=\delta_n^n$ and letting $n\to\infty$ we see that $v$ is a smooth solution of \eqref{2.5}. Besides, the inequality \eqref{2.11} says that $v$ satisfies the same bounds
\be\label{2.21}
0\leq v\leq -u_0,~~~\text{in}~\mathbb{R}^2,
\ee
which implies $v$ vanishes at infinity. Consequently, the proof of Theorem \ref{th2.1} is complete.

In the following we deal with the case $8\pi GN=1$ and all the centers of strings are identity which are not covered in the condition stated in Theorem \ref{th2.1}. It makes us to find a radially symmetric solution sufficiently.

\begin{theorem}\label{th2.2}
Under the condition $8\pi G N=1$ with $N\geq1$, if all the points $p_s$ are identical, that is to say, $p_s=p_0, s=1,2,\cdots, N$. Then the equation \eqref{1.31} has a solution that vanishes at infinity with the choice $g_0=2^{16\pi G}N$ and the solution is symmetry about $p_0$.
\end{theorem}
\begin{proof}

Without loss of generality, we assume that the single center of the $N$ strings: $p_s=p_0, s=1,2,\cdots, N$ is at the origin of $\mathbb{R}^2$. Set
\be\label{2.22}
a=8\pi G
\ee
and use $r=|x|$. Then \eqref{1.31} can be written as
\be\label{2.23}
\triangle u=g_0 r^{-2aN}\e^{au-\frac a m\e^{mu}}(\e^{mu}-1)+4\pi N\delta(x).
\ee
According to the radically symmetric of the solution, the equation \eqref{2.23} is reduced to a single one
\bea
u_{rr}+\frac 1 r u_r&=&g_0 r^{-2aN}\e^{au-\frac a m\e^{mu}}(\e^{mu}-1),\label{2.24}\\
\lim_{r\to 0}\frac{u(r)}{\ln r}&=&\lim_{r\to 0}r u_r=2N.\label{2.25}
\eea

We now give the new variables
\be\nn
t=\ln r,~~~~~U(t)=u(\e^t),
\ee
and use the condition $aN=1$, then the system \eqref{2.24}--\eqref{2.25} is equivalent to
\bea
U''&=&g_0\e^{a U-\frac a m \e^{m U}}(\e^{m U}-1),~~-\infty<t<\infty,\label{2.26}\\
\lim_{t\to -\infty}\frac{U(t)}{t}&=&\lim_{t\to -\infty}U'(t)=2N.\label{2.27}
\eea
Denote the right--hand side of the equation \eqref{2.26} as $h(U)$, then the integral form of system \eqref{2.26}--\eqref{2.27} is
\be\label{2.28}
U(t)=2Nt+\int_{-\infty}^t(t-\tau)h(U(\tau))\dd \tau,~~~t\in\mathbb{R}.
\ee
It is convenient to set $w=U-2Nt$, so we can rewrite \eqref{2.28} as
\be\label{2.29}
w(t)=\int_{-\infty}^t(t-\tau)h(2N\tau+w(\tau))\dd \tau,~~~t\in\mathbb{R}.
\ee
We are going to find a solution of the equation \eqref{2.29} which vanishes at $t=-\infty$. Denote the right--hand side of the equation \eqref{2.29} as $T(w)$, then we derive a fixed--point problem, $w=T(w)$. Given the following function space
\be\nn
\mathcal{X}=\left\{w\in C(-\infty, t_0]~\Big|~\lim_{t\to -\infty}w(t)=0,~~\sup_{t\leq t_0}|w(t)|\leq 1\right\},
\ee
where $-\infty<t_0<\infty$ is to be determined later. First, we show that $T$ maps from $\mathcal{X}$ into itself. In fact, for any $w\in\mathcal{X}$, we can choose $t_0$ properly to obtain
\bea\label{2.30}
|T(w)|&\leq& g_0\sup_{t\leq t_0}\int_{-\infty}^t|t-\tau|\e^{a(2N\tau+w(\tau))}\left|\e^{m(2N\tau+w(\tau))}-1\right|\dd \tau\nn\\
&\leq& g_0\sup_{t\leq t_0}\int_{-\infty}^t|t-\tau|\e^{a(2N\tau+1)}\dd \tau \leq 1.
\eea

Next, we show that $T$ is a contraction. For this purpose, in view of \eqref{2.26}, for any $U$ there holds
\be\label{2.31}
|h'(U)|=g_0\frac{\left|m\e^{(a+m)U}-a \e^{a U}(\e^{mU}-1)^2\right|}{\e^{\frac a m \e^{m U}}}\leq C_1 \e^{aU},
\ee
where $C_1$ is a positive constant depending on $g_0$ and $m$. Then, for any $w_1, w_2\in\mathcal{X}$, we have
\bea\label{2.32}
&&|T(w_1)-T(w_2)|\nn\\
=&&\left|\int_{-\infty}^t(t-\tau)h'\left(2N\tau+\widetilde{w}(\tau)\right)(w_1(\tau)-w_2(\tau))\dd \tau\right|\nn\\
\leq && C_1\sup_{t\leq t_0}|w_1(t)-w_2(t)|\int_{-\infty}^t(t_0-\tau)\e^{a(2N\tau+1)}\dd \tau,
\eea
where $\widetilde{w}(t)$ lies between $w_1(t)$ and $w_2(t)$. Therefore, when $t_0$ is properly chosen, $T: \mathcal{X}\to\mathcal{X}$ is a contraction. Consequently, we see that the system \eqref{2.26}--\eqref{2.27} has a negative solution in the neighborhood of $t=-\infty$ and $U(t)\to -\infty$ as $t\to -\infty$.

Besides, we can rewrite the right--hand side of \eqref{2.26} as
\be\label{2.33}
h(U)=g_0\e^{a(U-\frac 1 m \e^{m U})}(\e^{m U}-1)=-\frac{g_0}{a}\frac{\dd}{\dd U}\left[e^{a(U-\frac 1 m \e^{m U})}\right].
\ee
Then, multiplying the both side of \eqref{2.26} by $U'$ and integrating over $(-\infty, t)$, we have
\be\label{2.34}
(U')^2=4N^2-\frac{2g_0}{a}\e^{a(U-\frac 1 m \e^{m U})}\equiv F(U).
\ee
The critical point of the equation \eqref{2.34}, say $\hat{U}$ satisfies $F(\hat{U})=0$. In order to ensure uniqueness at the equilibrium $\hat{U}$, we are motivated to require that
\be\label{2.35}
F'(\hat{U})=2g_0 (\e^{m \hat{U}}-1)\e^{a(\hat{U}-\frac 1 m \e^{m \hat{U}})}=0.
\ee
Thus the only choice is  $\hat{U}=0$. As a consequence, inserting this result into $F(\hat{U})=0$ and use the condition $a N=1$, we can determine the parameter $g_0$ and
\be\label{2.36}
g_0=2a N^2 \e^{\frac a m}=2 N\e^{\frac{1}{m N}}.
\ee
In view of \eqref{2.27}, in the neighborhood of $t=-\infty$, we can rewrite \eqref{2.34} in the following explicit form
\be\label{2.37}
U'(t)=2N\sqrt{1-\e^{\frac 1 N(\frac 1 m+U-\frac 1 m\e^{m U})}}.
\ee
Since $U'(t)>0$ in \eqref{2.37} and the equilibrium $\hat{U}=0$ is unique. According to that $F(U)$ decreases in $U<0$ we can derive $U=U(t)$ solves \eqref{2.37} for $-\infty<t<\infty$ and $U(t)<0=\hat{U}$ for all $t$.

Besides, \eqref{2.37} can be written in the integral form
\be\label{2.38}
\int_{U(0)}^{U(t)}\frac{\dd U}{\sqrt{F(U)}}=t.
\ee
Thus, we see $U(t)\to\hat{U}=0$ as $t\to\infty$. In fact, since
\be\nn
F''(0)=2g_0 m\e^{-\frac a m}=4N m\e^{\frac{1}{mN}-\frac a m}=4Nm.
\ee
In view of \eqref{2.38}, we have the sharp estimate
\be\label{2.39}
|U|=O(\e^{-\sqrt{2Nm}t}),~~~t\to\infty.
\ee

Returning to the original variable $r=\e^t, u(r)=U(\ln r)$, we obtain the solution of \eqref{1.31}, and the decay estimate
\be\label{2.40}
|u(r)|=O(r^{-\sqrt{2Nm}}),~~~r\to\infty.
\ee
Note that \eqref{2.37} implies
\be\nn
\frac{U'(t)}{U(t)}\to\sqrt{\frac{F''(0)}{2}}=\sqrt{2Nm},~~~t\to\infty,
\ee
therefore, we have from \eqref{2.40} that
\be\label{2.41}
|u_r(r)|=O(r^{-(1+\sqrt{2Nm})}),~~~r\to\infty.
\ee
The theorem \ref{th2.2} is thus proven.
\end{proof}

So far we have proved the existence of cosmic string solutions. In the following, we will give the asymptotic properties for multistring solutions.

Choosing $r_0>0$ sufficiently large, such that
\be\nn
\{p_1, p_2, \cdots, p_N\}\subset B(r_0)=\{x\in\mathbb{R}^2 \mid |x|<r_0\},
\ee
then the equation \eqref{1.31} becomes
\be\label{2.42}
\triangle u=g_0\e^{a(u-{\frac{1}{m}\e^{mu}})}(\e^{mu}-1)\prod_{s=1}^N|x-p_s|^{-2a},~~~x\in\mathbb{R}^2\backslash\overline{B(r_0)}.
\ee

\begin{lemma}\label{d.3}
Suppose $a N<1$, then the solution of \eqref{2.42} holds the bound
\be\label{2.43}
-C_b|x|^{-b}\leq u(x)\leq 0,~~~|x|>r_0,
\ee
where $b$ is any positive constant and $C_b>0$ is a suitable constant depends on $b$. If $a N=1$ and there are at least two distant string centers, then \eqref{2.43} holds for $b=2$. If $a N=1$ and all the string centers are coinciding, then \eqref{2.43} holds for $b=\sqrt{2 N m}$.
\end{lemma}
\begin{proof}
First, we assume $a N<1$ and introduce the comparison function
\be\label{2.44}
w(x)=C|x|^{-b}.
\ee
Then
\be\label{2.45}
\triangle w=b^2r^{-2}w,~~~|x|=r>r_0.
\ee
In view of \eqref{2.42}, we have for $r_0>0$ sufficiently large that
\bea\label{2.46}
\triangle (u+w)&=&g_0\e^{a(u-{\frac{1}{m}\e^{mu}})}(\e^{mu}-1)\prod_{s=1}^N|x-p_s|^{-2a}+b^2r^{-2}w\nn\\
&=&g_0 m u \e^{a (u-\frac{1}{m}\e^{mu})}\e^{\xi m u}\prod_{s=1}^N|x-p_s|^{-2a}+b^2r^{-2}w\nn\\
&\leq &b^2r^{-2}(u+w),~~~|x|=r>r_0.
\eea
where $\xi\in [0, 1]$. For such fixed $r_0$, we may take $C>0$ in \eqref{2.44} large enough to make
\be\nn
(u(x)+w(x))\Big|_{|x|=r_0}\geq 0.
\ee
Hence, applying the maximum principle, we have $u\geq -w$ in $\mathbb{R}^2\backslash\overline{B(r_0)}$ as claimed.

If $a N=1$ and there are at least two distant string centers. By the definition of $u_0$, we know
\be\label{2.47}
-u_0(x)=\sum_{s=1}^N\ln\left(1+\frac{1}{|x-p_s|^2}\right)=O(r^{-2}),~~~|x|=r\to\infty,
\ee
combining with \eqref{2.21}, we conclude that
\be\label{2.48}
-C_2|x|^{-2}\leq u\leq 0,~~~|x|\to\infty.
\ee
If $a N=1$ and all the string centers are coinciding, the estimate follow from \eqref{2.40}.
\end{proof}

\begin{lemma}\label{d.4}
Suppose $a N\leq 1$ and there are at least two distant string centers. Then $\pa_j v\to 0$ as $|x|\to\infty, j=1, 2$.
\end{lemma}
\begin{proof}
In view of \eqref{2.21} and \eqref{2.47}, we know $v\in L^2(\mathbb{R}^2)$. Besides, Note that
\bea\label{2.50}
0\leq 1-\e^{m(u_0+v)}&\leq&1-\e^{m u_0}\nn\\
&=&1-\prod_{s=1}^N\left(\frac{|x-p_s|^2}{1+|x-p_s|^2}\right)^m=O(r^{-2}),~~~|x|=r\to\infty
\eea
and
\be\label{2.51}
\prod_{s=1}^N|x-p_s|^{2a}=O(r^{-2aN}),~~~|x|=r\to\infty,
\ee
we see the right--hand side of \eqref{2.5} belongs to $L^2(\mathbb{R}^2)$ as well. Thus,  using the $L^2$--estimates for \eqref{2.5}, we arrive at $v\in W^{2, 2}(\mathbb{R}^2)$.

Differentiating \eqref{2.5}, we have
\bea\label{2.52}
\triangle(\pa_j v)=&&g_0\e^{a((u_0+v)-{\frac{1}{m}\e^{m(u_0+v)}})}\left(m\e^{m(u_0+v)}-a(\e^{m(u_0+v)}-1)^2\right)\left(\prod_{s=1}^N|x-p_s|^{-2a}\right)(\pa_j v)\nn\\
&&+\frac{g_0}{a}\e^{a(v-{\frac{1}{m}\e^{m(u_0+v)}})}\left(m\e^{m(u_0+v)}-a(\e^{m(u_0+v)}-1)^2\right)\left(\prod_{s=1}^N|x-p_s|^{-2a}\right)(\pa_j \e^{a u_0})\nn\\
&&+g_0\e^{a((u_0+v)-{\frac{1}{m}\e^{m(u_0+v)}})}(\e^{m(u_0+v)}-1)\pa_j\left(\prod_{s=1}^N|x-p_s|^{-2a}\right)+(\pa_j g).
\eea
Clearly, $\pa_j g\in L^{2}(\mathbb{R}^2)$. Because of
\be\label{2.53}
\pa_j\left(\prod_{s=1}^N|x-p_s|^{-2a}\right)=O(r^{-(2aN+1)})
\ee
and
\be\label{2.54}
\pa_j\e^{a u_0}=\pa_j\left(\prod_{s=1}^N\frac{|x-p_s|^2}{1+|x-p_s|^2}\right)^a=O(r^{-1}),
\ee
we see the right--hand side of \eqref{2.52} lies in $L^{2}(\mathbb{R}^2)$. Using the $L^2$--estimates again, we arrive at the conclusion $\pa_j v\in W^{2,2}(\mathbb{R}^2)$. Therefore, $\pa_j v\to 0$ as $|x|\to\infty$.
\end{proof}

Using the definition of $u_0$ and lemma \ref{d.4}, we derive that $|\nabla u|\to 0$ as $|x|\to\infty$.

\begin{lemma}\label{d.5}
For the solution $u$ of \eqref{2.42}, there establish
\be\label{2.55}
|\nabla u|^2\leq C_b|x|^{-b},~~~|x|>r_0,
\ee
where $C_b>0$ is a constant depending on $b$. If $aN<1$, then $b>0$ is an arbitrary constant. If $aN=1$ and there are at least two distant string centers, then $b=3$.
\end{lemma}
\begin{proof}
Differentiating \eqref{2.42} gives us
\bea\label{2.56}
\triangle (\pa_j u)=&&g_0\e^{a(u-{\frac{1}{m}\e^{mu}})}\prod_{s=1}^N|x-p_s|^{-2a}\left(m\e^{m u}-a(\e^{mu}-1)^2\right)(\pa_j u)\nn\\
&&+g_0\e^{a(u-{\frac{1}{m}\e^{mu}})}(\e^{m u}-1)\pa_j\left(\prod_{s=1}^N|x-p_s|^{-2a}\right).
\eea
In view of lemma \ref{d.4} and \eqref{2.53}, at infinity, there holds,
\be\label{2.57}
(\e^{m u}-1)\pa_j\left(\prod_{s=1}^N|x-p_s|^{-2a}\right)=O(r^{-b_1}),~~~\forall b_1>0,
\ee
when $a N<1$, or
\be\label{2.58}
(\e^{m u}-1)\pa_j\left(\prod_{s=1}^N|x-p_s|^{-2a}\right)=O(r^{-5}),
\ee
when $a N=1$ and there are at least two distant string centers.

Set $h=|\nabla u|^2$, according to \eqref{2.56}, when $|x|>r_0$, there holds
\bea\label{2.59}
\triangle h&=&\sum_{j=1}^{n}\triangle(\pa_j u)^2=\sum_{i,j=1}^{n}\pa_i\pa_i(\pa_j u)^2\geq\sum_{i,j=1}^{n}2(\pa_j u)(\pa_i\pa_i\pa_j u)\nn\\
&&+2\sum_{j=1}^{n}\pa_j u\left\{g_0\e^{a(u-{\frac{1}{m}\e^{mu}})}(\e^{m u}-1)\pa_j\left(\prod_{s=1}^N|x-p_s|^{-2a}\right)\right\}\nn\\
&=&2g_0\e^{a(u-{\frac{1}{m}\e^{mu}})}\prod_{s=1}^N|x-p_s|^{-2a}\left(m\e^{m u}-a(\e^{mu}-1)^2\right)h\nn\\
&&+2g_0\e^{a(u-{\frac{1}{m}\e^{mu}})}(\e^{m u}-1)\nabla u\cdot\nabla\left(\prod_{s=1}^N|x-p_s|^{-2a}\right)\nn\\
&\geq&2g_0\e^{a(u-{\frac{1}{m}\e^{mu}})}\prod_{s=1}^N|x-p_s|^{-2a}\left(m\e^{m u}-a(\e^{mu}-1)^2\right)h+q(x),
\eea
where $q(x)$ enjoys the same decay estimates as \eqref{2.57} and \eqref{2.58}.

Taking $w$ as in \eqref{2.44}, using \eqref{2.45} and \eqref{2.59}, we find, when $|x|=r>r_0$ there holds
\be\label{2.60}
\triangle(h-w)\geq 2g_0\e^{a(u-{\frac{1}{m}\e^{mu}})}\prod_{s=1}^N|x-p_s|^{-2a}\left(m\e^{m u}-a(\e^{mu}-1)^2\right)h-b^2r^{-2}w+q(x).
\ee
Assume
\bea
b_1&>&2+b~~\text{when}~~a N<1,\nn\\
5&\geq&2+b~~\text{when}~~a N=1,\nn
\eea
then we may choose a suitable $C>0$ in \eqref{2.44}, such that
\be\label{2.61}
q(x)>-b^2r^{-2}w(x),~~~|x|=r>r_0.
\ee
If $a N<1$, we may also assume
\be\label{2.62}
g_0\e^{a(u-{\frac{1}{m}\e^{mu}})}\prod_{s=1}^N|x-p_s|^{-2a}\left(m\e^{m u}-a(\e^{mu}-1)^2\right)>b^2r^{-2},~~~|x|=r>r_0;
\ee
if $a N=1$, we can make $g_0$ sufficiently large, so that \eqref{2.62} is still established.

Inserting \eqref{2.61} and \eqref{2.62} into \eqref{2.60}, we obtain
\be\label{2.63}
\triangle(h-w)\geq 2b^2r^{-2}(h-w),~~~|x|=r>r_0.
\ee
Therefore, we can let the constant $C$ in \eqref{2.44} be large that
\be\label{2.64}
\left(h(x)-w(x)\right)\big|_{|x|=r_0}\leq 0.
\ee
Hence, applying the maximum principle we arrive at $h\leq w$ for $|x|>r_0$.
\end{proof}



\medskip

{\bf Acknowledgments.}
This work was supported by NSFC-12101197 and China Postdoctoral Science Foundation. No. 2022M721022.


\begin{thebibliography}{99}
\bibitem{Ag}
S. Agmon, A. Douglis, and L. Nirenberg, Estimates near the boundary of solutions of elliptic partial differential equations satisfying general boundary conditions. I, {\em Comm. Pure Appl. Math.} {\bf12}, 623--727 (1959).

\bibitem{Be}
L. Bers, F. John, and M. Schechter, {\em Partial Differential Equations}, Amer. Math. Soc., Providence, 1964.

\bibitem{Br}
R. H. Brandenberger, Cosmic strings and the large--scale structure of the universe, {\em Phys. Scr.} {\bf T36}, 114--126 (1991).

\bibitem{Caf}
L. Caffarelli, Y. Yang, Vortex condensation in the Chern--Simons--Higgs model: an existence theorem, {\em Commun. Math. Phys.} {\bf168} 321--336(1995).

\bibitem{Cha}
D. Chae, O. Yu. Imanuvilov, The existence of nontopological multivortex solutions in the relativistic self--dual Chern--Simons theory, {\em Commun. Math. Phys.} {\bf215} 119--142(2000).

\bibitem{Co}
A. Comtet, G. W. Gibbons, Bogomol'nyi bounds for cosmic strings, {\em Nucl. Phys.} B {\bf299} 719--733 (1988).

\bibitem{Gr}
P. A. Griffiths, J. Harris, {\em Principle of Algebraic Geometry}, Wiley, New York, 1978.

\bibitem{Gu}
S. B. Gudnason, Nineteen vortex equations and integrability, {\em J. Phys.} A {\bf55} 405401(2022).

\bibitem{Ho}
John van der Hoek, M. A. Lohe, Vortex properties in first--and second--order formulations of abelian gauge theories, {\em J. Math. Phys.} {\bf25}, 154--160(1984).

\bibitem{Ki}
T.W.B. Kibble, Topology of cosmic domains and strings, {\em J. Phys. A: Math. Gen.} {\bf9} 1387 (1976).

\bibitem{Kib}
T.W.B. Kibble, Some implications of a cosmological phase transition, {\em Phys. Reports} {\bf67} 183--199 (1980).

\bibitem{Le}
P. S. Letelier, Multiple cosmic strings, {\em Class. Quantum Grav.} {\bf4} 75--77(1987).

\bibitem{Lo}
M. A. Lohe, Generalized noninteracting vortices, {\em Phys. Rev.} D {\bf23}, 10(1981).

\bibitem{Loh}
M. A. Lohe, John van der Hoek, Existence and uniqueness of generalized vortices, {\em J. Math. Phys.} {\bf24}, 148--153 (1983).

\bibitem{Ma}
N. S. Manton, Five vortex equations, {\em J. Phys.} A {\bf50} 125403(2017).

\bibitem{Ni}
W.-M. Ni, On the elliptic equation $\triangle u+K(x)u^{(n+2)/(n-2)}=0$, its generalizations, and applications in geometry, {\em Indiana U. Math. J.} {\bf 31}, 493--529 (1982).

\bibitem{Spru}
J. Spruck, Y. Yang, Topological solutions in the self--dual Chern--Simons theory: existence and approximation, {\em Ann. Inst. H. Poincar\'{e}-- Anal. non lineaire} {\bf12} 75--97(1995).

\bibitem{Spr}
J. Spruck and Y. Yang, The existence of non--topological solitons in the self--dual Chern--Simons theory, {\em Commun. Math. Phys.} {\bf149} 361--376(1992).

\bibitem{Tar}
G. Tarantello, Multiple condensate solutions for the Chern--Simons--Higgs theory, {\em J. Math. Phys.} {\bf37} 3769--3796(1996).

\bibitem{Vi}
A. Vilenkin, Cosmological density fluctuations produced by vacuum strings, {\em Phys. Rev. Lett.} {\bf46}, 1496 (1981).

\bibitem{Wi}
E. Witten, Superconducting strings, {\em Nucl. Phys.} B {\bf249}, 557--592 (1985).

\bibitem{YaO}
Y. Yang, Obstructions to the existence of static cosmic strings in an Abelian Higgs model, {\em Phys. Rev. Lett.} {\bf72} 10--13(1994).

\bibitem{YaP}
Y. Yang, Prescribing topological defects for the coupled Einstein and Abelian Higgs equations, {\em Commun. Math. Phys.} {\bf170} 541--582(1995).

\bibitem{Ze}
Ya. B. Zel'dovich, Cosmological fluctuations produced near a singularity, {\em Mon. Not. Roy. Astron. Soc.} {\bf192}, 663--667 (1980).

\end{thebibliography}
\end{document}